\documentclass[preprint]{elsarticle}
\usepackage{graphicx} 

\usepackage{natbib}
\usepackage{pifont}
\usepackage{geometry}
\usepackage{fleqn}
\usepackage{graphicx}
\usepackage{hyperref}
\usepackage{amsmath, amsthm, amssymb, amsfonts,mathrsfs}
\usepackage{braket}
\usepackage{color}
\usepackage{bbold}

\newtheorem{defi}{Definition}
\newtheorem{lemma}[defi]{Lemma}
\newtheorem{coro}[defi]{Corollary}
\newtheorem{prop}[defi]{Proposition}
\newtheorem{theo}[defi]{Theorem}

\theoremstyle{definition}

\def\tr{{\rm tr}}
\def\PP{\mathbb{P}}
\def\hh{\mathfrak{h}}

\title{Purification of quantum trajectories in infinite dimensions}

\author[1]{Federico Girotti}
\ead[1]{federico.girotti@polimi.it}
\affiliation[1]{organization={Dipartimento di Matematica, Politecnico di Milano},
addressline={Via Edoardo Bonardi 9},
postcode={20133},
city={Milano},
country={Italy}}
\cortext[1]{federico.girotti@polimi.it}

\author[1]{Alessandro Vitale}
\ead[1]{alessandro3.vitale@mail.polimi.it}

\begin{document}
\begin{abstract}
In this work we exhibit a class of examples that show that the characterization of purification of quantum trajectories in terms of `dark' subspaces that was proved for finite dimensional systems (\cite{MK06,BP03}) fails to hold in infinite dimensional ones. Moreover, we prove that the new phenomenon emerging in our class of models and preventing purification to happen is the only new possibility that emerges in infinite dimensional systems. Our proof strategy points out that the emergence of new phenomena in infinite dimensional systems is due to the fact that the set of orthogonal projections is not sequentially compact. Having in mind this insight, we are able to prove that the finite dimensional extends to a class of infinite dimensional models.
\end{abstract}

\maketitle

\makeatletter
\def\ps@pprintTitle{%
  \let\@oddhead\@empty
  \let\@evenhead\@empty
  \let\@oddfoot\@empty
  \let\@evenfoot\@oddfoot
}
\makeatother

\section{Introduction}

Quantum trajectories are stochastic processes describing the evolution of quantum systems undergoing repeated indirect measurements. They were first introduced in the description of continuously monitored quantum systems (see for instance \cite{Da76,BLP82,BL85,Be89,Be89b}) and as useful tools for computations in open quantum systems (see \cite{Da14} and references therein). Remarkably, quantum trajectories were also employed in the theoretical description of the experiments conducted by Serge Haroche's group (\cite{GBD07,HR06}) and of the results obtained by David Wineland (\cite{LBMW03,Wi13}). In this work we will focus on quantum trajectories in discrete time, which can be seen as discretizations of continuous time models (\cite{Pe10}).

\bigskip If the measurement is perfect, i.e. there is no information flowing into the system and all the information leaking from the system is observed, the set of pure states is a closed set for the dynamics: if the systems starts in a pure state, then it is in a pure state at every time. A natural question is under which conditions the set of pure states is also attractive, in the sense that the state of the system tends to `purify' almost surely for long times and for every initial state; aside from its own interest, there are several motivations for studying purification: first of all, assuming  that the system purifies almost surely and that it satisfies another irreducibility assumption is key in proving several results concerning quantum trajectories, such as uniqueness of the invariant measure, convergence to the unique invariant measure and limit theorems for empirical averages of a wide class of functionals (see \cite{BFPP19,BFP23,BHP25}). Moreover, measurement driven purification is a promising way of preparing pure states and the study of the dependence of the purification time on the system size for different measurement strengths has been investigated in several recent works (\cite{GH20, DLLNZ24} and references therein).

\bigskip If the system has finitely many degrees of freedom, a characterization of purification was found in \cite{MK06} (see \cite{BP03} for the same result in continuous time and \cite{BFPP19} for an alternative proof): purification occurs unless the dynamics hits upon a family of `dark' subspaces, i.e. subspaces from which there is no leak of information. In this work we make a first step in trying to understand better purification in infinite dimensional systems; we present a class of models for which purification does not hold even if there are not any `dark' subspaces, showing that, in general, the characterization for finite dimensional systems does no longer hold in infinitely many dimensions. The rationale behind the class of examples is that the dynamics moves along a family of subspaces that become closer and closer to `dark' subspaces. The main result of this work is showing that this is a necessary condition for purification to fail (Theorem \ref{theo:main}); the proof makes use of the infinite dimensional version of Nielsen's inequalities (\cite{Ni01}) and adapts the techniques employed in \cite{BP03} to the infinite dimensional and discrete time case. As a side result, we provide an alternative proof of the characterization of purification in terms of `dark' subspaces in discrete time with respect to the one in \cite{MK06} (Corollary \ref{coro:fd}) and in \cite{BFPP19} (Proposition 2.2). Moreover, in Proposition \ref{prop:sm} we show that such a characterization still holds for a class of models with an increasing sequence of finite dimensional subspaces where trajectories `remain confined'. 

\bigskip We point out that our results are not the first ones regarding purification in infinite dimensional settings: in \cite{BBP24}, the authors show that purification holds for a particular model describing an atom maser, building on the technique used in \cite{BFPP19}. Moreover, \cite{Li10} considers quantum trajectories satisfying a technical assumption regarding both the initial state and the measurement; adapting the techniques used in \cite{MK06}, the author shows that, under such assumptions, the failure of almost sure purification implies that there exists a non trivial `dark' subspace.

The structure of the paper is the following: Section \ref{sec:pre} introduces the notation and the problem studied, in Section \ref{sec:main} we present the class of examples showing that the characterization of purification in terms of `dark' subspaces does not hold in the infinite dimensional case and we prove the main theorem (Theorem \ref{theo:main}) concerning the necessary condition for the failure of purification of quantum trajectories. Finally, Section \ref{sec:ds} shows how the finite dimensional result follows from Theorem \ref{theo:main} and provides a class of infinite dimensional models for which the characterization of purification in terms of `dark' subspaces still holds true. 

\section{Preliminaries and notation}\label{sec:pre}

Let us consider a quantum system described by a separable Hilbert space $\hh$. We will identify quantum states with their densities matrices, which are positive semidefinite trace class operators with unit trace.

\bigskip \noindent \textbf{Single measurement.} We assume that the system interacts with an ancillary one, described by the separable Hilbert space $\hh_a$ and initially prepared in a given pure state $\ket{\chi}\bra{\chi}$, according to a certain unitary operator $U:\hh \otimes \hh_a \rightarrow \hh\otimes \hh_a$. After the interaction, the measurement corresponding to a certain orthonormal basis $\{\ket{i}\}_{i \in I}$ is performed on the ancillary system. Therefore, according to Born's rule, if the system is initially prepared in the state $\rho$, one observes the outcome $i$ with probability $\tr(\rho a_i^* a_i)$ and the state of the system conditional to such outcome is given by
$$\frac{a_i \rho a_i^*}{\tr(\rho a_i^* a_i)},$$
where $a_i:=\bra{i}U\ket{\chi}$ are called Kraus operators. Notice that $\sum_{i \in I}a_i^*a_i=\mathbf{1}$, where the convergence of the series has to be intended in the strong operator topology. The described quantum measurement is said to be \textit{perfect} because the measurement on the ancillary system is nondegenerate and because the ancilla is initially prepared in a pure state.

\bigskip \noindent\textbf{Repeated measurements.} Let us assume to repeat the previous procedure infinitely many times, always considering a new ancilla at every step; the stochastic process describing the state of the system conditional to the sequence of outputs is called quantum trajectory and will be the central object in our investigation. Let us define the process formally: let $\Omega$ be the set of infinite sequences with entries in $I$ and let us consider cylinder sets, i.e. those sets of the form
$$\Lambda_{i_1,\dots, i_m}:=\{\omega=(\omega_1,\dots, \omega_n, \dots)  \in \Omega: \omega_k=i_k, \, k=1,\dots, m\}, \quad m \in \mathbb{N}^*, \, i_1,\dots, i_m \in I.$$
We denote by ${\cal F}_m$ the $\sigma$-field generated by cylinder sets corresponding to outcomes sequences of length $m$ (${\cal F}_0:=\{\emptyset, \Omega\}$) and by ${\cal F}$ the $\sigma$-field generated by all cylinder sets. For every initial state of the system $\rho$ there exists a unique probability measure $\PP_\rho$ on $(\Omega, {\cal F})$ satisfying
$$\PP_\rho(\Lambda_{i_1,\dots,i_n})=\tr(\rho a_{i_1}^* \cdots a_{i_n}^* a_{i_n} \cdots a_{i_1}).$$
This can be seen via Kolmogoroff extension theorem. We will denote by $\mathbb{E}_\rho$ the corresponding expected value. We are now ready to introduce the quantum trajectory $(\rho_n)_{n \geq 0}$ as the adapted stochastic process on $(\Omega, ({\cal F}_n)_{n \geq 0}, {\cal F})$ given by
$$\rho_0(\omega):=\rho, \quad \rho_n(\omega):=\frac{a_{\omega_n} \cdots a_{\omega_1} \rho a_{\omega_1}^* \cdots a_{\omega_n}^*}{\tr(\rho a_{\omega_1}^* \cdots a_{\omega_n}^* a_{\omega_1} \cdots a_{\omega_n})}.$$
Notice that $(\rho_n)_{n \geq 0}$ is a Markov process with transitions governed by the following kernel:
\[
\rho_{n+1}=\frac{a_i \rho_n a_i^*}{\tr(a_i \rho_n a_i^*)} \text{ with probability } \pi_{i,n}:=\tr(\rho_n a_i^*a_i), \quad i \in I.
\]

\bigskip \noindent \textbf{Purification.} A state $\rho$ is called pure if it is a rank one projector, i.e. if there exists a vector $\ket{\psi} \in \hh$ such that $\rho=\ket{\psi}\bra{\psi}$; pure states correspond to extremal elements of the set of states. Notice that perfect measurements preserve the set of pure states: indeed, if $\rho=\ket{\psi}\bra{\psi}$, then for every $\omega \in \Omega$ one has
$$\rho_n(\omega)=\frac{\ket{a_{\omega_n} \cdots a_{\omega_1}\psi}\bra{a_{\omega_n} \cdots a_{\omega_1}\psi}}{\|a_{\omega_n} \cdots a_{\omega_1}\psi\|^2}.$$

As we already mentioned in the introduction, we want to study in more detail those cases in which not only the set of pure states is closed for the dynamics, but also the measurement tends to purify every initial state for long times. In order to give a mathematical definition of this phenomenon, we need to introduce a measure of how far a state is from being pure. We will use the linear entropy of a state, i.e. the quantity $$g(\rho):=\tr(\rho(\mathbf{1}-\rho))=1-\tr(\rho^2).$$

$\mathbf{1}$ denotes the identity operator on $\hh$. One can easily see that for every state $0 \leq g(\rho) < 1$ and that it is equal to $0$ if and only if the state is pure. We are now ready to define precisely what we mean by purification of quantum trajectories due to the back-action of the measurement.
\begin{defi}[Pur]
We say that the measurement asymptotically purifies the trajectories if for every initial state $\rho $ the following holds true:
\begin{equation} \label{eq:puri} 
\lim_{n \rightarrow +\infty} g(\rho_n)=0 \quad \PP_\rho-a.s.
\end{equation}
\end{defi}

Notice that our definition is equivalent to the one in \cite{MK06} due to the fact that H{\"o}lder inequality ensures that $\tr(\rho^2)\leq \tr(\rho^m)^{1/(m-1)}$ for every $m \geq 2$.

If the Hilbert space $\hh$ is finite dimensional, the following equivalent characterization of (Pur) was found in \cite[Corollary 4]{MK06} (see \cite[Theorem 2.1]{BP03} for similar result in continuous time and \cite[Proposition 2.2]{BFPP19} for an alternative proof):

\begin{center} 
\textit{there does not exist any projection $p$ whose support has dimension bigger or equal than $2$ and such that for all $n \in \mathbb{N}^*$ and $i_1,\dots, i_n \in I$ there exists a nonnegative constant $\lambda_{i_1,\dots, i_n}$ satisfying}
\end{center}
\begin{equation} \label{eq:PI} pa^*_{i_n}\cdots a^*_{i_1}a_{i_1} \cdots a_{i_n}p=\lambda_{i_1,\dots, i_n}p.\end{equation}

The support of such a projection is usually referred to as a `dark' subspace, in that there is no information leaking into the ancillas. Indeed, any state supported on the support of $p$ produces the same outcome probability distribution and there is no way to discriminate two different states from the outcome sequence. Moreover, Eq. \eqref{eq:PI} means that every product of $a_i$'s acts as a multiple of an isometry when restricted to the support of $p$, therefore, if $\rho_0=p\rho_0p$, then $\rho_0$ and $\rho_n$ share the same eigenvalues for every time $n$ and there may be a change only in the eigenvectors: if $\rho_0$ is not a pure state, the same holds for $\rho_n$ at every time.

\section{Counterexample and main result} \label{sec:main}

In this section we exhibit a class of examples that shows that (Pur) fails to be equivalent to the condition in Eq. \eqref{eq:PI} for infinite dimensional systems. More precisely, we will show that for such class of models (Pur) fails, but there does not exist any `dark' subspace. After identifying the new phenomenon appearing in infinite dimensional systems and causing the failure of purification, we show that this is the only possibility outside of the ones already observed in the finite dimensional setting. 

\bigskip Let us consider the Hilbert space $\    \hh=\ell^2(\mathbb{N}^*)$ and let $(\ket{e_k})_{k \geq 1}$ be the orthonormal basis given by $\ket{e_k}(m):=\delta_{km}$. Consider the quantum trajectories associated to the following Kraus operators:
\[
a_1 = \sum_{k \geq 1} \sqrt{\alpha(k)} \ket{e_{k+1}}\bra{e_k}, \quad 
a_2 = \sum_{k \geq 1} \sqrt{1 - \alpha(k)} \ket{e_{k+1}}\bra{e_k}
\]
where the function $\alpha$ takes values in $[0,1]$.

\begin{prop} \label{prop:ce}
If
$$\forall k \in \mathbb{N}, \quad \alpha(k+1) > \alpha(k),
$$
then, for every initial state of the form $\rho=\gamma_0 \ket{e_1}\bra{e_1} + \delta_0 \ket{e_2}\bra{e_2}$, with $\gamma_0,\delta_0 \in (0,1)$ and $\gamma_0+\delta_0=1$, the following statement holds true:
$$\PP_\rho\left (\lim_{n \rightarrow +\infty}g(\rho_n)=0\right )=0.$$

Moreover, if $\alpha(n)=c-(n+z)^{-1}$ for some $c \in (0,1)$ and $z \in (c^{-1}-1,+\infty)$, then there exists no orthogonal projection $p$ with support of dimension at least $2$ such that for every $n \in \mathbb{N}^*$ and collection of indices $i_1,\dots, i_n \in \{1,2\}$, there exists a nonnegative number $\lambda_{i_1,\dots, i_n}$ satisfying
$$pa_{i_1}^* \cdots a_{i_n}^* a_{i_n} \cdots a_{i_1}p =\lambda_{i_1,\dots, i_n}p.$$
\end{prop}

The explanation of why any initial state of the form $\gamma_0 \ket{e_1}\bra{e_1} + \delta_0 \ket{e_2}\bra{e_2}$ does not purify, even if there are not any `dark' subspaces, is quite simple. Let us consider the sequence of orthogonal projections given by $$p_k:=\ket{e_k}\bra{e_k}+\ket{e_{k+1}}\bra{e_{k+1}}, \quad k \geq 1.$$
Notice that, for the class of initial states that we are considering, the support of $\rho_n$ is contained in the support of $p_n$. Thanks to the strict monotonicity of $\alpha$, any product of Kraus operators $a_{i_m} \cdots a_{i_1}$ does not act as a multiple of a partial isometry on the support of any $p_k$; however, this becomes true in the limit for $n \rightarrow +\infty$ due to the fact that $\alpha(k)\xrightarrow[k\rightarrow+\infty]{}c$. This is a phenomenon that can only appear in infinite dimensional systems; moreover, the following result shows that this is the only cause for purification to fail.

\begin{theo} \label{theo:main}
If (Pur) does not hold, then there exists a sequence of projections $\{p_n\}_{n\geq 0}$ such that
\begin{enumerate}
    \item the dimension of the support of each $p_n$ is at least $2$,
    \item for every $m \in \mathbb{N}^*$ and $i_1,\dots, i_m \in I$, there exists a sequence of positive numbers $\lambda_{i_1,\dots, i_m,n}$ satisfying
$$\lim_{n \rightarrow +\infty}\|p_n a_{i_1}^*\cdots a_{i_m}^*a_{i_m} \cdots a_{i_1}p_n-\lambda_{i_1,\dots, i_m,n}p_n\|_\infty=0.$$
\end{enumerate} 
\end{theo}

The proof employs the ideas and tools of \cite[Theorem 2.1]{BP03}, however facing some extra difficulties due to the infinite dimensional setting. For instance, one phenomenon that might take place in infinite dimensions and could prevent purification of quantum trajectories is the flattening of the spectral measure of the conditional state, i.e. the lack of tightness of the probability densities on $\mathbb{N}^*$ given by $(\mu_k(\rho_n))_{k \geq 1}$, where for every state $\rho$, we denote by $(\mu_k(\rho))_{k \geq 1}$ the sequence of its eigenvalues (with multiplicities) in non-decreasing order. However, the following Lemma contains some concentration inequalities that show that this scenario is not possible. The proof of the following Lemma is based on an extension of Eq. (28) in \cite{Ni01} to the infinite dimensional setting. We point out that Eq. (28) in \cite{Ni01} was also used in the proof of \cite[Corollary 5]{MK06} for different purposes.

\begin{lemma} \label{lem:conc}
For every $N \geq 1$, let us define
$S^N_n:=\sum_{k \geq 1}^{N}\mu_k(\rho_n)$; there exists a random variable $0 \leq S^N_\infty \leq 1$ such that
$$\lim_{n \rightarrow +\infty}S_n^N=S_\infty^N \quad \PP_\rho-\text{a.s.}.$$
Moreover, for every $\epsilon>0$, there exists $N \geq 1$ such that for every $n \geq 0$ and $\gamma \in (0,1)$
    $$\PP_\rho\left (S_\infty^N<\gamma\right )\leq \frac{\epsilon}{1-\gamma}.$$
\end{lemma}

\section{Bringing `dark' subspaces back} \label{sec:ds}

In the previous section we saw that the characterization of purification in terms of `dark' subspaces given in Eq. \eqref{eq:PI} may fail due to the fact that in infinite dimensional models the supports of $p_n$'s can escape to infinity. This intuition is helpful in determining models for which this cannot happen: for instance, the proof of Theorem 2.1 in \cite{BP03} and Theorem 1 \cite{MK06} relied on the fact that if the system is finite dimensional the set of projections is sequentially compact and one can show that any projection which is a limit point of $(p_n)_{n \geq 1}$ is the projection onto a `dark' subspace. In this section we detail how the finite dimensional result follows from Theorem \ref{theo:main}, we determine another class of models for which purification can be characterized in terms of `dark' subspaces and we exhibit several examples that belong to such class.

\begin{coro}[Finite dimensional case] \label{coro:fd}
Let us assume that $\hh$ is finite dimensional. The following statements are equivalent:
\begin{enumerate}
\item (Pur) does not hold,
\item there exists a projection $p$ whose support has dimension bigger or equal than $2$ and such that for all $n \in \mathbb{N}^*$ and $i_1,\dots, i_n \in I$
$$pa^*_{i_n}\cdots a^*_{i_1}a_{i_1} \cdots a_{i_n}p=\lambda_{i_1,\dots, i_n}p.$$
\end{enumerate}
\end{coro}
\begin{proof}
The implication $2. \Rightarrow 1.$ is trivial, since it is enough to take any mixed state $\rho_0$ supported in $p$.

\bigskip $1. \Rightarrow 2.$ Let us consider a sequence of orthogonal projections as in the statement of Theorem \ref{theo:main}; since the set of orthogonal projections is sequentially compact in finite dimensional Hilbert spaces, the sequence $\{p_{n}\}_{n \geq 0}$ admits at least an accumulation point $p$ which satisfies the statement in Eq. \eqref{eq:PI}.
\end{proof}

\begin{prop} \label{prop:sm}
Let us assume that there exists an increasing sequence of finite dimensional projections $(s_m)_{m \geq 1}$ such that
\begin{enumerate}
    \item \label{item:compl} $\bigvee_{m \geq 1}s_m=\mathbf{1},$
    \item for every $m \geq 0$ and for every initial state $\rho$, $(X_n^m:=\tr(\rho_n s_m))_{n \geq 0}$ is a $(({\cal F}_n)_{n \geq 0}, \PP_\rho)$-submartingale. 
\end{enumerate}

Then, the following statements are equivalent:
\begin{enumerate}
\item (Pur) does not hold,
\item there exists a projection $p$ whose support has dimension bigger or equal than $2$ and such that for all $n \in \mathbb{N}^*$ and $i_1,\dots, i_n \in I$ there exists a nonnegative constant $\lambda_{i_1,\dots,i_n}$ satisfying
$$pa^*_{i_n}\cdots a^*_{i_1}a_{i_1} \cdots a_{i_n}p=\lambda_{i_1,\dots, i_n}p.$$
\end{enumerate}
\end{prop}

$\bigvee_{m \geq 1}s_m$ denotes the smallest orthogonal projection $s$ such that $s \geq s_m$ for every $m \geq 1$. Item 2. in the hypotheses of Proposition \ref{prop:sm} might look quite abstract and hard to check; however, there is a simple criterion on Kraus operators that ensures it holds. Let us consider a increasing sequence of finite dimensional projections $(s_m)_{m \geq 1}$ satisfying item 1. in Proposition \ref{prop:sm} and assume that Kraus operators satisfy
\begin{equation} \label{eq:Kcond}a_i s_m=s_ma_is_m, \quad \forall m \geq 1,\, \forall i \in I.\end{equation}
This means that Kraus operators are upper triangular in the decomposition of the Hilbert space induced by $(s_m)_{m \geq 1}$, i.e. if we consider $\hh=\bigoplus_{m \geq 1}\hh_m$ where $\hh_m$ is the range of $s_m -s_{m-1}$ ($s_0:=0$), then the block structure of Kraus operators is of the form
\[
\left(\begin{array}{c|c|c}
   \mbox{\normalfont\Large\bfseries *}
  &  \mbox{\normalfont\Large\bfseries *} & \mbox{\normalfont\Large\bfseries ...}\\
\hline
  \mbox{\normalfont\Large\bfseries 0}
  & \mbox{\normalfont\Large\bfseries *} & \mbox{\normalfont\Large\bfseries ...}\\
  \hline
  \mbox{\normalfont\Large\bfseries ...}
  & \mbox{\normalfont\Large\bfseries ...} & \mbox{\normalfont\Large\bfseries ...}
\end{array}\right).
\]

In this case, one can easily check that assumption 2. in Proposition \ref{prop:sm} is satisfied: indeed, for every initial state $\rho$ one has
\[\begin{split}\mathbb{E}_\rho[X_{n+1}^{m}-X_{n}^m|{\cal F}_n]&=\sum_{i \in I}\tr\left (\frac{a_i \rho_n a_i^*}{\tr(a_i \rho_n a_i^*)}s_m\right )\tr(a_i \rho_n a_i^*)-\tr(\rho_n s_m)  \\
&=\tr\left (\rho_n\left (\sum_{i \in I}a_i^* s_m a_i -s_m\right ) \right ) \geq 0, \end{split}\]
where the last inequality follows from the fact that condition in Eq. \eqref{eq:Kcond} is equivalent to $ \sum_{i \in I}a_i^* s_m a_i \geq s_m$ (see Section 3 in \cite{CP16}).

\section{Conclusions}
The goal of this work is to take a first step in the understanding of purification of quantum trajectories in infinite dimensional systems. The main contributions of this investigation are showing that the characterization of purification in terms of `dark' subspaces ceases to hold true in infinite dimension and identifying the only possible new phenomenon appearing in infinite dimensional systems that can prevent purification to happen, which is a sequence of subspaces connected by the dynamics and which become closer and closer to a `dark subspace. An interesting open problem is whether this is also a sufficient condition (possibly adding some further conditions). Finally, we proved that the characterization in terms of `dark' subspaces keeps being true for a class of models even in infinite dimensional systems; we believe that this result can be improved, proving it for a wider set of physical systems.
\section{Acknowledgments}
The authors are grateful to Franco Fagnola for suggesting the topic for A.V.'s master thesis and for many useful discussions. F.G. is a member of GNAMPA-INdAM and  acknowledges the support of the MUR grant ``Dipartimento di Eccellenza 2023--2027'' of Dipartimento di Matematica, Politecnico di Milano and ``Centro Nazionale di ricerca in HPC, Big Data and Quantum Computing''. 
\appendix

\section{Counterexample and main result}

\begin{proof}[Proof of Proposition \ref{prop:ce}]

\textbf{Failure of purification.} Consider the state at time $n$: 

$$\rho_n=\gamma_n \ket{e_{n+1}}\bra{e_{n+1}}+\delta_{n}\ket{e_{n+2}}\bra{e_{n+2}}.$$
Let us assume that at time $n+1$ one observe the outcome $1$, then
\begin{equation}\label{gamma1}\gamma_{n+1} \leq \gamma_n \text{ and } \delta_{n+1} \geq \delta_n.\end{equation}

Indeed, we recall that in this case one has
$$\rho_{n+1}=\frac{\alpha(n+1)\gamma_{n}\ket{e_{n+2}}\bra{e_{n+2}}+\alpha(n+2)\delta_n\ket{e_{n+3}}\bra{e_{n+3}}}{\alpha(n+1)\gamma_{n}+\alpha(n+2)\delta_n}.$$

Thus,

\[\begin{split}
\gamma_{n+1}-\gamma_n&=\frac{\alpha(n+1)\gamma_{n}}{\alpha(n+1)\gamma_{n}+\alpha(n+2)\delta_n}-\gamma_n =\frac{\alpha(n+1)(\gamma_n-\gamma_n^2)-\alpha(n+2)\delta_n\gamma_n}{\alpha(n+1)\gamma_{n}+\alpha(n+2)\delta_n}
\\
&=\frac{\alpha(n+1)\gamma_n\delta_n-\alpha(n+2)\delta_n\gamma_n}{\alpha(n+1)\gamma_{n}+\alpha(n+2)\delta_n}\\
&=\frac{(\alpha(n+1)-\alpha(n+2))\gamma_n\delta_n}{\alpha(n+1)\gamma_{n}+\alpha(n+2)\delta_n} \leq 0,
\end{split}\]
which proves the inequalities in \eqref{gamma1} (the statement for $\delta_n$ and $\delta_{n+1}$ follows from the fact that $\delta_n=1-\gamma_n$ and $\delta_{n+1}=1-\gamma_{n+1}$).

On the other hand, if at time $n+1$ one observes the outcome $2$, then the reverse inequalities holds:
\begin{equation}\label{delta1}
    \gamma_{n+1} \geq \gamma_n \text{ and } \delta_{n+1} \leq \delta_n.
\end{equation}

Indeed,

$$\rho_{n+1}=\frac{\left(1-\alpha(n+1)\right)\gamma_{n}\ket{e_{n+2}}\bra{e_{n+2}}+\left(1-\alpha(n+2)\right)\delta_n\ket{e_{n+3}}\bra{e_{n+3}}}{\left(1-\alpha(n+1)\right)\gamma_{n}+\left(1-\alpha(n+2)\right)\delta_n}$$

and

\[\begin{split}
 \gamma_{n+1}-\gamma_n&=\frac{\left(1-\alpha(n+1)\right)\gamma_n}{\left(1-\alpha(n+1)\right)\gamma_{n}+\left(1-\alpha(n+2)\right)\delta_n}-\gamma_n\\
 &=
\frac{(1-\alpha(n+1))\gamma_n\delta_n -(1-\alpha(n+2))\gamma_n\delta_n}{\left(1-\alpha(n+1)\right)\gamma_{n}+\left(1-\alpha(n+2)\right)\delta_n}\\
&=\frac{(\alpha(n+2)-\alpha(n+1))\gamma_n\delta_n }{\left(1-\alpha(n+1)\right)\gamma_{n}+\left(1-\alpha(n+2)\right)\delta_n} \geq 0.
\end{split}\]

Therefore this proves the inequalities in \eqref{delta1}.

Consider any measurement trajectory $\omega \in \Omega$; we will denote by $\omega_{n]}$ its restriction to the first $n$-outcomes. Notice that

\begin{equation}\label{ineq1}
\begin{split}
&\gamma_{n+1}(\omega) \delta_{n+1}(\omega)=\gamma_{n+1}(\omega_{n+1]})\delta_{n+1}(\omega_{n+1}])\geq \gamma_{n+1}(\omega_{n]},1)\delta_{n+1}(\omega_{n]},2)= \\ &=\frac{\alpha(n+1)\gamma_n(\omega)}{\alpha(n+1)\gamma_n(\omega)+\alpha(n+2)\delta_n(\omega)}\cdot \frac{(1-\alpha(n+2))\delta_n(\omega)}{(1-\alpha(n+1))\gamma_n(\omega)+(1-\alpha(n+2))\delta_n(\omega)}.  
\end{split}\end{equation}

Now consider the denominator in \eqref{ineq1}:\\
\[\begin{split}
&(\alpha(n+1)\gamma_n(\omega)+\alpha(n+2)\delta_n(\omega))\cdot ((1-\alpha(n+1))\gamma_n(\omega)+(1-\alpha(n+2))\delta_n(\omega)) \leq \\
&\alpha(n+2) (1-\alpha(n+1)).
\end{split}\]

Therefore,
$$\gamma_{n+1}(\omega)\delta_{n+1}(\omega)\geq \frac{\alpha(n+1)(1-\alpha(n+2))}{\alpha(n+2)(1-\alpha(n+1))}\gamma_n(\omega)\delta_n(\omega).$$

Iterating this inequality yields:
$$\gamma_{n+1}(\omega_{n+1]})\delta_{n+1}(\omega_{n+1]})\geq \frac{\alpha(1) (1-\alpha(n+2))}{\alpha(n+2) (1-\alpha(1))}\gamma_0\delta_0
$$
and we get that for every $\omega \in \Omega$ 
\begin{equation}
\liminf_{n \rightarrow +\infty}\gamma_{n+1}(\omega)\delta_{n+1}(\omega) \geq \frac{\alpha(1) (1-c)}{c (1-\alpha(1))}\gamma_0 \delta_0 >0.    
\end{equation}
The statement follows from the fact that $$g(\rho_n)=1-\gamma_n^2-\delta_n^2=2\gamma_n\delta_n.$$

\bigskip \noindent \textbf{Absence of `dark' subspaces.} We will prove the statement by contradiction. Let us assume that such a projection exists; then, there exist two orthonormal vectors $v,w$ that satisfy the following conditions: for every $n \in \mathbb{N}^*$ and collection of indices $i_1,\dots, i_n \in \{1,2\}$ 
\begin{equation}\label{eq:isocond} \langle v, a_{i_1}^*\cdots a_{i_n}^* a_{i_n} \cdots a_{i_1}v \rangle=\langle w, a_{i_1}^*\cdots a_{i_n}^* a_{i_n} \cdots a_{i_1}w \rangle, \quad  \langle v, a_{i_1}^*\cdots a_{i_n}^* a_{i_n} \cdots a_{i_1}w \rangle=0.
\end{equation}
Let us rephrase the conditions in Eq. \eqref{eq:isocond} in a more manageable way. Before doing that, we need to introduce some more notation; first of all, let us consider the following family of spaces of complex valued sequences:
$$\ell^p(\mathbb{N}^*):=\left \{(f(n))_{n \geq 1}:\sum_{n \geq 1}|f(n)|^p<+\infty\right \}.$$
We denote by $S$ the right shift operator, i.e. $S$ maps the sequence $(f(n))_{n \geq 1} \in \ell^p(\mathbb{N}^*)$ into the sequence
$$\ell^p(\mathbb{N}^*) \ni Sf(n)=\begin{cases} 0 \text{ if }n=1\\
f(n-1) \text{ if } n \geq 2 \end{cases} $$
for every $p \in [1,+\infty]$. Moreover, we denote by $S^*$ the left shift, which maps the sequence $(f(n))_{n \geq 1} \in \ell^p(\mathbb{N}^*)$ into the sequence $(f(n+1))_{n \geq 1} \in \ell^p(\mathbb{N}^*)$ (we chose the notation $S^*$ to recall that $S^*$ is the adjoint of $S$). Given any bounded function $g:\mathbb{N}^* \rightarrow \mathbb{R}$, $g(N)$ denotes the operator that maps $(f(n))_{n \geq 1} \in \ell^p(\mathbb{N}^*)$ into the sequence $(g(n)\cdot f(n))_{n \geq 1}\in\ell^p(\mathbb{N}^*)$ for every $p \in [1,+\infty]$; we remark that $g(N)S=Sg(N+1)$, $g(N)^*=g(N)$ (where, again, $^*$ denotes the adjoint) and $g(N)h(N)=h(N)g(N)$ for any other $h:\mathbb{N}^* \rightarrow \mathbb{R}$. Notice that
$$a_1=S \sqrt{\alpha(N)}, \quad a_2=S \sqrt{1-\alpha(N)},$$
therefore for every $n \in \mathbb{N}^*$, $i_1,\dots, i_n \in \{1,2\}$ one has
$$a_{i_1}^*\cdots a_{i_n}^* a_{i_n}\cdots a_{i_1}=\sqrt{\beta_{i_1}(N)}S^*\cdots \sqrt{\beta_{i_n}(N)}S^* S \sqrt{\beta_{i_n}(N)}\cdots S \sqrt{\beta_{i_1}(N)}=\beta_{i_1}(N)\cdots \beta_{i_n}(N+n),$$
where $\beta_1(n):=\alpha(n)$ and $\beta_2(n):=1-\alpha(n)$. Let us denote $\beta_{i_1,\dots, i_n}:=(\beta_{i_1,\dots, i_n}(k):=\beta_{i_1}(k)\cdots \beta_{i_n}(k+n))_{k \geq 1} \in \ell^{\infty}(\mathbb{N}^*)$, then conditions in Eq. \eqref{eq:isocond} can be rephrased as
$$\sum_{k \geq 1} \beta_{i_1,\dots, i_n}(k)\cdot (|v(k)|^2-|w(k)|^2)=0, \quad \sum_{k \geq 1} \beta_{i_1,\dots, i_n}(k)\cdot \overline{v}(k)\cdot w(k)=0.$$
Therefore $d:=(d(k):=(|v(k)|^2-|w(k)|^2))_{k \geq 1}$ and $f:=(f(k):=\overline{v}(k)\cdot w(k))_{k \geq 1}$ are vectors in $\ell^1(\mathbb{N}^*)$ which are annihilated by all continuous linear functionals in
$${\cal H}:={\rm span}_{\mathbb{C}}\{\underline{1},\beta_{i_1,\dots, i_n}, \, i_1, \dots, i_n \in \{1,2\}, \, n \in \mathbb{N}^*\} \subseteq \ell^{\infty}(\mathbb{N}^*),$$
where $\underline{1}$ denotes the sequence identically equal to $1$.

If we show that the only element in $\ell^1(\mathbb{N}^*)$ which is in the kernel of all the elements in ${\cal H}$ is zero, we are done: indeed, this would imply that
$$d=f=0 \text{ or equivalently } v=w,$$
which is a contradiction because we assumed $v$ and $w$ to be orthonormal. First of all, notice that $\underline{1}$ and $ S^{*k} \alpha$ for every $k \geq 1$ belong to ${\cal H}$, where $\alpha:=(\alpha(n))_{n \geq 1}.$ Since $\ell^1(\mathbb{N}^*)\subseteq \ell^2(\mathbb{N}^*)$, any nontrivial vector $x \in \ell^1(\mathbb{N}^*)$ that is in the kernel of $\underline{1}$ and $ S^{*k} \alpha$ for every $k \geq 1$ is also in the kernel of the operator $A:\ell^2(\mathbb{N}^*) \rightarrow \ell^2(\mathbb{N}^*)$ determined by
$$\bra{k}A\ket{m}=(k+m+z)^{-1}, \quad k,m \geq 1.$$
Indeed, if we see $A$ as an infinite matrix, its $k$-th row is given by $c\underline{1}-S^{*k}\alpha$. However, Corollary 7.4 in \cite{Si21} ensures that $A$ has trivial kernel and we are done.
\end{proof}

\begin{proof}[Proof of Lemma \ref{lem:conc}.]
    First of all, an extension of Equation (28) in \cite{Ni01} to the infinite dimensional setting shows that for every $N\geq 1$, $S^N_n:=\sum_{k=1}^{N}\mu_k(\rho_n)$ is a $(({\cal F}_n)_{n \geq 0},\PP_\rho)$-submartingale. The proof is based on Ky Fan's maximum principle (Lemma II.15 in \cite{DPTG16}) and can be found in \cite{Li10} (see Proposition 8.14).

    Therefore, Doob's convergence theorem implies tha there exists a random variable $S_\infty^N$ such that
    $$\lim_{n \rightarrow +\infty}S_n^N=S_\infty^N \quad \PP_\rho -\text{a.s.}.$$
    
   We set $S_\infty^N:=0$ when the limits does not exist. Moreover, using the submartingale property and Markov inequality one has

    \[\begin{split}\PP_\rho\left (S_\infty^N <\gamma\right )&=\PP_\rho\left (1-S_\infty^N >1-\gamma\right ) \leq \frac{1-\mathbb{E}_\rho\left [ S_\infty^N\right ]}{1-\gamma} \\
    &\leq \frac{1-\mathbb{E}_\rho[S_0^N]}{1-\gamma}=\frac{1-\sum_{k=1}^{N}\mu_k(\rho)}{1-\gamma}.\end{split}\]

    Since $\sum_{k=1}^{+\infty}\mu_k(\rho)=1$, for every $\epsilon >0$, there exists $N\geq 1$ such that
    $$1-\sum_{k=1}^{N}\mu_k(\rho)\leq \epsilon$$
    and we are done.
    \end{proof}

\begin{proof}[Proof of Theorem \ref{theo:main}]

    First of all, we need to prove a technical lemma.
    
   \begin{lemma} \label{lem:A} If there exists a state $\rho$ such that
$$\PP_\rho\left (\lim_{n \rightarrow +\infty}g(\rho_n)=0\right )<1,$$
then there exists an event $A \in {\cal F}$ such that
\begin{enumerate}
\item $\PP_\rho(A)>0$,
\item there exists a constant $0 <c \leq 1$ such that for every $\omega \in A$ the following limit exists and it satisfies:
$$\lim_{n \rightarrow +\infty}g(\rho_n(\omega))\geq c ,$$
\item there exist two constants $a$ and $b$ such that for every $\omega \in A$ one has 
$$0 < a \leq \liminf_{n \rightarrow +\infty}\mu_2(\rho_n(\omega)) \leq \limsup_{n\rightarrow +\infty}\mu_1(\rho_n(\omega)) \leq b <1,$$
\item there exists an increasing sequence of natural numbers $(n_j)_{j\geq 1}$ such that for every every $\omega \in A$ one has
$$\lim_{j\rightarrow +\infty}\sum_{\underline{i} \in I^p}\pi_{\underline{i},n_j}(\omega)\tr \left ( \left (\sqrt{\rho_{n_j}(\omega)}\left (\frac{a_{\underline{i}}^*a_{\underline{i}}}{\pi_{\underline{i},n_j}(\omega)}-\mathbf{1}\right ) \sqrt{\rho_{n_j}(\omega)}\right )^2 \right )=0, \quad \forall p \geq 1.$$
\end{enumerate}
\end{lemma}
\begin{proof}
Let $\rho$ be a state such that trajectories do not purify $\PP_\rho$ a.s.. Let us consider the stochastic process given by the linear entropy and its mean, i.e.
    $$g_n:=g(\rho_n), \quad G_n:=\mathbb{E}_\rho[g_n].$$

    For every $p \geq 1$ and $n\geq 1$, we can write
    \begin{equation} \label{eq:deco}g_{n}=\sum_{k=0}^{p-1}g_k-\sum_{k=1}^{p-1}g_{n-k}-\sum_{k=1}^{n-p+1}h^p_k+m^p_n,\end{equation}
    where
    $$m_n^p:=\sum_{k=p}^{n}(g_k-\mathbb{E}_\rho[g_k|{\cal F}_{k-p}])$$
    and
    $$h^p_k=\mathbb{E}_\rho[g_{k-1}-g_{k-1+p}|{\cal F}_{k-1}].$$

    Notice that $m_n^p$ is a centered random variable and it is a $(({\cal F}_n)_{n \geq 0}, \PP_\rho)$-martingale for $p=1$. Moreover, $h_k^p$ is ${\cal F}_{k-1}$-measurable, therefore the decomposition for $p=1$ is the Doob decomposition.
    
    If we show that $h_k^p \geq 0$ almost surelya, we prove that $g_n$ is a supermartingale. Before doing that, let us introduce some notation: given any sequence of outcomes $\underline{i}=(i_1,\dots, i_p) \in I^p$, let us introduce the following notation:
    $$a_{\underline{i}}^*a_{\underline{i}}:=a_{i_1}^* \cdots a_{i_p}^* a_{i_p} \cdots a_{i_1}, \quad \pi_{\underline{i},k}:=\tr(\rho_{k}a_{\underline{i}}^*a_{\underline{i}}), \quad \rho_{\underline{i},k}=\frac{a_{\underline{i}} \rho_{k}a_{\underline{i}}^*}{\tr(\rho_{k}a_{\underline{i}}^*a_{\underline{i}})}.$$

    One has
    \[\begin{split}
        h_k^p&=\sum_{\underline{i} \in I^p}\pi_{\underline{i},k-1}\tr(\rho_{\underline{i},k-1}^2)-\tr(\rho_{k-1}^2)\\
        &=\sum_{\underline{i} \in I^p}(\pi_{\underline{i},k-1}\tr(\rho_{\underline{i},k-1}^2)-2\tr(\rho_{k-1}^2 a_{\underline{i}}^*a_{\underline{i}})+\pi_{\underline{i},k-1}\tr(\rho_{k-1}^2))\\
        &=\sum_{i \in I^p}\pi_{\underline{i},k-1}\tr \left (\rho_{k-1} \left (\frac{a^*_{\underline{i}}a_{\underline{i}}}{\pi_{\underline{i},k-1}}-\mathbf{1}\right ) \rho_{k-1}\left (\frac{a^*_{\underline{i}}a_{\underline{i}}}{\pi_{\underline{i},k-1}}-\mathbf{1}\right ) \right )=\\
        &=\sum_{\underline{i} \in I^p}\pi_{\underline{i},k-1}\tr \left ( \left (\sqrt{\rho_{k-1}}\left (\frac{a_{\underline{i}}^*a_{\underline{i}}}{\pi_{\underline{i},k-1}}-\mathbf{1}\right ) \sqrt{\rho_{k-1}}\right )^2 \right )\geq 0.
    \end{split}\]
    
    Since $0 \leq g_n\leq 1$, Doob's convergence theorem ensures that there exists $\lim_{n \rightarrow +\infty}g_n=:g_\infty$ $\PP_\rho$-a.s. and in $L^1(\PP_\rho)$, therefore $G_\infty:=\lim_{n\rightarrow +\infty}G_n=\mathbb{E}_\rho[g_\infty]$. We set $g_\infty:=0$ when the limit does not exists.
    
    Using the decomposition of $g_n$ in Eq. \eqref{eq:deco}, we can write
    $$G_n=\sum_{k=0}^{p-1}G_k-\sum_{k=1}^{p-1}G_{n-k}-\sum_{k=1}^{n}H^p_k, \text{ where } H^p_k:=\mathbb{E}_\rho[h^p_k].$$
    Since $0 \leq G_n \leq 1$, $H^p_k \geq 0$ and $|\sum_{k=0}^{p-1}G_k-\sum_{k=1}^{p-1}G_{n-k}| \leq 2p$, the following must hold true
    $\lim_{k\rightarrow +\infty}H^p_k=0.$
    Therefore, modulo passing to a subsequence $k_j$, one has
    $$\lim_{j\rightarrow +\infty}h^p_{k_j}=0, \, \forall p \geq 1 \quad \PP_\rho \, a.s.$$
    and item 3. follows denoting $n_j:=k_j-1.$

    Since the trajectory does not purify, there must exist a strictly positive constant $c$ such that $\PP_\rho(A_c)>0$, where $A_c:=\{g_\infty\geq c>0\}$. Let us consider any $\omega \in A$; for any $b$ such that $c >1-b^2>0$, there exists $\overline{n}(\omega)>0$ such that for every $n \geq \overline{n}(\omega)$ one has
    $$g_n(\omega)=1-\sum_{k=1}^{+\infty}\mu_k(\rho_n(\omega))^2\geq 1-b^2.$$
    Notice that
    \[ b^2 \geq \sum_{k=1}^{+\infty}\mu_k(\rho_n(\omega))^2 \geq \mu_1(\rho_n(\omega))^2,\]
    therefore $\mu_1(\rho_n(\omega))\leq b<1.$
    Moreover, for every $a^\prime$ such that $1-b>a^\prime>0$ consider $N$ such that
    $$\PP_\rho(\{S_\infty^N\geq b+a^\prime \}   \cap A_c)>0$$
    (this is possible thanks to Lemma \ref{lem:conc}) and define $A:=\{S_\infty^N\geq b+a^\prime\}   \cap A_c.$ For every $\omega \in A$, one has that for every $a^{\prime\prime}$ such that $a^\prime >a^{\prime\prime} >0$ there exists $K(\omega)\geq 0$ such that for every $n \geq K(\omega),$
    $$b+a^{\prime\prime} \leq S_n^N(\omega)=\mu_1(\rho_n(\omega))+\sum_{k =2}^N\mu_k(\rho_n(\omega)) \leq b+\sum_{k =2}^N\mu_k(\rho_n(\omega)); $$
    therefore,
    $$\mu_2(\rho_n(\omega)) \geq \frac{1}{N-1}\sum_{k =2}^N\mu_k(\rho_n(\omega)) \geq \frac{a^{\prime\prime}}{N-1}=:a$$
    and we are done.
    \end{proof}

    Let us consider a trajectory $\rho_n(\omega)$ for $\omega \in A$; notice that all the properties listed in items 2.-4. of Lemma \ref{lem:A} are inherited by subsequences of $(n_j)_{j \geq 1}$.
    
    Let us denote by $\mu^1_{j} \geq \mu_j^2$ the two biggest eigenvalues of $\rho_{n_j}(\omega)$, which we called $\mu_1(\rho_{n_j}(\omega))$ and $\mu_2(\rho_{n_j}(\omega))$ so far; moreover, let $q^1_j$, $q^2_j$ be two orthogonal one-dimensional projections such that $\rho_{n_j}(\omega)q^i_{j}=\mu^i_j q^i_j$ for $i=1,2$. We introduce the notation $p_j:=q^1_j+q^2_j$.
    
 We denote by $I^p_0$ the set of indices $\underline{i} \in I^p$ such that $\lim_{j \rightarrow +\infty}\pi_{\underline{i},n_j}(\omega)=0$. Without loss of generality, we can assume that the following holds true:
 
 \begin{equation}\label{eq:picond}
 \text{either }\underline{i} \in I^p_0 \text{ or there exists $\epsilon_{\underline{i}}>0$ such that $\pi_{\underline{i},n_j}(\omega) \geq \epsilon_i$.}
 \end{equation}
 
 Indeed, once we label the indices in $\bigcup_{p \geq 1}I^p\setminus I_0^p$ as $\underline{i}_1,\dots, \underline{i}_k, \dots $ we can pick a subsequence of $n_j$, that we denote by $l_{11},\dots, l_{1k},\dots,$ such that $\pi_{\underline{i}_1,l_{1j}}(\omega) \geq \epsilon_1$ for some $\epsilon_1>0$. Then we pass to $\underline{i}_2$: if $\lim_{j\rightarrow +\infty}\pi_{\underline{i}_2,l_{1j}}(\omega)=0$, we add $\underline{i}_2$ to $I^p_0$, where $p$ is the length of $\underline{i}_2$; otherwise we can extract a subsequence $l_{2j}$ from $l_{1j}$ such that $\pi_{\underline{i}_2,l_{2,j}}(\omega) \geq \epsilon_2$ for some $\epsilon_2>0$. We can do this for every $\underline{i} \in \bigcup_{p \geq 1}I^p \setminus I_0^p$. The subsequence $\{l_{kk}\}_{k \geq 1}\subseteq \{n_j\}$ has the property that we required.

 We are now ready to prove the statement. If $\underline{i} \in \bigcup_{p \geq 1}I^p_0$ , then
 \[\begin{split}   
 0&=\lim_{j \rightarrow +\infty}\pi_{\underline{i},n_j}(\omega)=\lim_{j \rightarrow +\infty}\tr(\rho_{n_j}(\omega) a_{\underline{i}}^*a_{\underline{i}})\geq \limsup_{j \rightarrow +\infty}\mu^2_{j}\tr(p_{j} a_{\underline{i}}^*a_{\underline{i}})\geq a \limsup_{j \rightarrow +\infty}\tr(p_{j} a_{\underline{i}}^*a_{\underline{i}})\\
 &=a \limsup_{j \rightarrow +\infty}\tr(p_{j} a_{\underline{i}}^*a_{\underline{i}})=
 a \limsup_{j \rightarrow +\infty}\|p_{j} a_{\underline{i}}^*a_{\underline{i}}p_j\|_1,\end{split}\]
 therefore
 $$\lim_{j\rightarrow +\infty}\|p_ja_{\underline{i}}^*a_{\underline{i}}p_j-\lambda_{\underline{i},j}p_j\|_\infty=0$$
 for $\lambda_{\underline{i},j}=0$ (since $p_j$ is finite, the two norms are comparable).
 On the other hand, if $\underline{i} \in I^p\setminus I^p_0$ for some $p \geq 1$, then
 $$\lim_{j \rightarrow +\infty}\tr \left (\rho_{n_j}(\omega) \left (a^*_{\underline{i}}a_{\underline{i}}-\pi_{\underline{i},k_j}(\omega)\mathbf{1}\right ) \rho_{n_j}(\omega)\left (a^*_{\underline{i}}a_{\underline{i}}-\pi_{\underline{i},n_j}(\omega)\mathbf{1}\right )\right )=0$$
 and, if we define $\widetilde{\rho}_{n_j}(\omega):=p_j\rho_{n_j}(\omega) p_j$, one has
 \[\begin{split}&\tr \left (\rho_{n_j}(\omega) \left (a^*_{\underline{i}}a_{\underline{i}}-\pi_{\underline{i},k_j}(\omega)\mathbf{1}\right ) \rho_{n_j}(\omega)\left (a^*_{\underline{i}}a_{\underline{i}}-\pi_{\underline{i},n_j}(\omega)\mathbf{1}\right )\right )\\
 &\geq \tr \left (\widetilde{\rho}_{n_j}(\omega) \left (a^*_{\underline{i}}a_{\underline{i}}-\pi_{\underline{i},n_j}(\omega)\mathbf{1}\right ) \widetilde{\rho}_{n_j}(\omega)\left (a^*_{\underline{i}}a_{\underline{i}}-\pi_{\underline{i},n_j}(\omega)\mathbf{1}\right )\right )\\
 &\geq (\mu_{j}^2)^2\|p_j(a^*_{\underline{i}}a_{\underline{i}}-\pi_{\underline{i},n_j}\mathbf{1} )p_j\|_2^2\end{split}\]
 and we are done (again using the fact that, since $p_j$ is finite, $\|p_j(a^*_{\underline{i}}a_{\underline{i}}-\pi_{\underline{i},n_j}\mathbf{1} )p_j\|_2$ and $\|p_j(a^*_{\underline{i}}a_{\underline{i}}-\pi_{\underline{i},n_j}\mathbf{1} )p_j\|_\infty$ are comparable). 
\end{proof}

\section{Bringing `dark' subspaces back}

\begin{proof}[Proof of Proposition \ref{prop:sm}]
The implication $2. \Rightarrow 1.$ is trivial.

\bigskip Notice that for every $n,m \geq 0$, one has $0 \leq X_n^m \leq 1$, therefore by Doob's convergence theorem there exists a sequence of random variables $(X^m_\infty)_{m \geq 0}$ such that
$$\lim_{n \rightarrow +\infty} X_n^m=X_\infty^m \quad \PP_\rho-\text{a.s..}$$
Moreover, by the fact that $s_m \uparrow \mathbf{1}$, one has that
$X_n^m \uparrow 1$ $\PP_\rho$-a.s.; this convergence passes to $X_\infty^m$ as well. Therefore
$$\PP_\rho \left ( B \right )=1, \text{ where } B:=\{\forall \delta \in (0,1), \, \exists M_\delta, N_\delta \geq 1 : \, \forall n \geq N_\delta, \, X_n^M \geq 1-\delta\}$$
and we can repeat the proof of the implication $1. \Rightarrow 2.$ of Theorem \ref{theo:main} considering a trajectory corresponding to $\omega \in A \cap B$, where $A$ is the event appearing in the statement of Lemma \ref{lem:A}.

In order to conclude, it is enough to show that that there exists a projection $p$ such that $p_j \xrightarrow[j \rightarrow +\infty]{\|\cdot \|_\infty}p$. We recall that the space of compact linear operators endowed with the uniform norm is a separable Banach space and that its dual is isometric to the space of trace class operators endowed with the trace norm. By Banach-Alaouglu's theorem, without loss of generality, we can assume that $(q_j^1)_{j\geq 0}$ and $(q_j^2)_{j \geq 0}$ converge in the ${\rm w}^*$-topology to two nonnegative trace class operators $q^1$ and $q^2$ (we can reduce to this situation passing to subsequences). Moreover, for every $\delta$, we can pick $M_\delta:=M_\delta(\omega)$ and $J_\delta(\omega)$ such that for every $j \geq J_\delta:=J_\delta(\omega)$,
$$0 <a^\prime \leq \mu_j^2 \leq \mu_j^1 \leq b^\prime <1, \text{ with } a^\prime <a \text{ and } b^\prime >b$$
and 
\[\begin{split}
    \mu_j^1 \tr(q_j^1 s_{M_\delta})+\mu_2^1 \tr(q_j^2 s_{M_\delta})+1-\mu_j^1-\mu_j^2&\geq\mu_j^1 \tr(q_j^1 s_{M_\delta})+\mu_2^1 \tr(q_j^2 s_{M_\delta})+\tr(\rho_{k_j}(\omega)-\widetilde{\rho}_{k_j}(\omega) s_{M_\delta})\\
    &=\tr(\rho_{k_j}(\omega)s_{M_\delta})\geq 1-\delta.
    \end{split}\]
Therefore, if we define $\gamma_j:=\mu_j^1/(\mu_j^1+\mu_j^2)$, we get that
$$\gamma_j \tr(q_j^1 s_{M_\delta})+(1-\gamma_j) \tr(q_j^2 s_{M_\delta}) \geq 1-\frac{\delta}{\mu_j^1+\mu_j^2} \geq 1-\frac{\delta}{2a^\prime}.$$
Since $1 \geq \tr(q_j^1 s_{M_\delta}), \tr(q_j^2 s_{M_\delta})$, one has
\[
    \min\{\gamma_j\tr(q_j^1 s_{M_\delta})+(1-\gamma_j),\gamma_j+(1-\gamma_j)\tr(q_j^2 s_{M_\delta})\} \geq \gamma_j \tr(q_j^1 s_{M_\delta})+(1-\gamma_j) \tr(q_j^2 s_{M_\delta})\geq 1-\frac{\delta}{2a^\prime}.\]
Therefore
$$\tr(q_j^1 s_{M_\delta})\geq 1-\frac{\delta}{2a^\prime\gamma_j}\geq 1-\frac{\delta\cdot b}{2a^{\prime 2}}, \quad \tr(q_j^2 s_{M_\delta})\geq 1-\frac{\delta}{2a^{\prime }(1-\gamma_j)}\geq 1-\frac{\delta\cdot b}{2a^{\prime 2}},$$
where we used the fact that $a \leq \mu_j^2 \leq \mu_j^1 \leq b.$

Summing up and with a slight abuse of notation, we can state in the following way what we just deduced: for every $\delta^\prime \in (0,1)$, there exist $M_{\delta^\prime}:=M_{\delta^\prime}(\omega), \, J_{\delta^\prime}:=J_{\delta^\prime}(\omega) \geq 1$ such that for every $j \geq J_{\delta^\prime}$, one has
$$\min\{\tr(q_j^1 s_{M_{\delta^\prime}}), \tr(q_j^2 s_{M_{\delta^\prime}})\} \geq 1-\delta^\prime.$$
The previous statement ensures a tightness property for the sequences $(q^{i}_j)_{j \geq 0}$, for $i=1,2$, therefore one can easily see that
$$q_j^i \xrightarrow[j\rightarrow +\infty]{ \|\cdot \|_1}q^i,$$

where $q^1$ and $q^2$ are orthogonal rank one projections (we recall that the set of pure states is a closed set in the topology induced by the trace class norm). Therefore, $p:=q^1+q^2$ is a rank two projection and $p_j \xrightarrow[j \rightarrow +\infty]{}p$ both in trace class and uniform norm and we are done.
\end{proof}

\bibliographystyle{abbrv}
\bibliography{biblio}

\end{document}